\documentclass[headsepline,10pt,a4paper]{amsart}
\usepackage{amsmath,amssymb,amsfonts,amsthm,exscale,calc}
\usepackage[latin1]{inputenc}
\usepackage{latexsym,textcomp}
\usepackage{graphicx}
\usepackage{parskip}
\usepackage{floatflt}
\usepackage{picins}
\usepackage{dsfont}
\usepackage{hyperref}
\usepackage{vmargin}

\newtheorem*{Thm}{Theorem}
\newtheorem{lemma}{Lemma}
\newtheorem{prop}{Proposition}

\newtheorem*{Cor}{Corollary}
\theoremstyle{definition}
\newtheorem*{Remark}{Remark}
\newtheorem*{definition}{Definition}

 \newlength\headseptemp

\newcommand{\Hm}[1]{\leavevmode{\marginpar{\tiny%
$\hbox to 0mm{\hspace*{-0.5mm}$\leftarrow$\hss}%
\vcenter{\vrule depth 0.1mm height 0.1mm width \the\marginparwidth}%
\hbox to 0mm{\hss$\rightarrow$\hspace*{-0.5mm}}$\\\relax\raggedright
#1}}}

\newcommand{\T}{{\mathbb T}}
\newcommand{\Z}{{\mathbb Z}}
\newcommand{\R}{{\mathbb R}}
\newcommand{\C}{{\mathbb C}}
\newcommand{\h}{{\mathbb H}}
\newcommand{\N}{{\mathbb N}}
\newcommand{\EE}{{\mathbb E}}
\newcommand{\PP}{{\mathbb P}}

\newcommand{\A}{{\mathcal A}}

\newcommand{\W}{{\mathcal W}}

\newcommand{\Lp}{{\Delta}}

\newcommand{\Leb}{\mathrm{Leb}}

\newcommand{\ka}{{\kappa}}
\newcommand{\al}{{\alpha}}
\newcommand{\be}{{\beta}}
\newcommand{\de}{{\delta}}
\newcommand{\Gm}{{\Gamma}}
\newcommand{\gm}{{\gamma}}
\newcommand{\ph}{{\varphi}}
\newcommand{\lm}{{\lambda}}
\newcommand{\te}{{\theta}}
\newcommand{\Te}{{\Theta}}
\newcommand{\eps}{{\varepsilon}}

\newcommand{\si}{{\sigma}}

\newcommand{\ap}[1]{\left( #1\right)}
\newcommand{\ab}[1]{\left( #1\right)}

\newcommand{\as}[1]{\langle #1\rangle}

\newcommand{\ov}[1]{\overline{#1}}
\newcommand{\ow}[1]{\widetilde{#1}}

\newcommand{\mo}[1]{\left\vert #1\right\vert}

\let\Im\undefined

\DeclareMathOperator{\Im}{Im}

\begin{document}
\title[AC spectrum for multi-type Galton Watson trees]{Absolutely continuous spectrum for multi-type Galton Watson trees}
\author{Matthias Keller}
\address{Mathematisches Institut, Friedrich Schiller Universit\"at Jena,
  D-07743 Jena, Germany, m.keller@uni-jena.de, http://www.analysis-lenz.uni-jena.de/Team/Matthias+Keller.html}
\subjclass[2000]{}

\begin{abstract}We consider multi-type Galton Watson trees that are close to a tree of finite cone type in distribution. Moreover, we impose that each vertex has at least one forward neighbor. Then, we show that the spectrum of the Laplace operator exhibits almost surely a purely absolutely continuous component which is included in the absolutely continuous spectrum of the tree of finite cone type.
\end{abstract}

\maketitle
\section{Introduction and main result}

After the seminal work of Klein \cite{Kl1,Kl3} there has been a lot of effort in the recent years to develop various techniques to show preservation of absolutely continuous spectrum for random operators on tree like graphs,  see \cite{ASW1,AW,FHH,FHS2,FHS3,Hal,Kel,KLW2,KlS1,KlS2}. While most of the work is concerned with diagonal  perturbations by a random potential, we consider a randomization of the geometry.

In particular, we study the absolutely continuous spectrum of the discrete Laplace operator on a multi-type Galton Watson tree. We show that parts of the spectrum are almost surely purely absolutely continuous, whenever the  distribution is close to a deterministic one. Trees associated to deterministic distributions are called trees of finite cone type or periodic trees. The spectrum of the Laplacian on such deterministic trees consists of  finitely many bands of  purely absolutely continuous spectrum, see \cite{Kel,KLW1}.
An important assumption on the distribution of the random trees is that each vertex has at least one forward neighbor. Without this assumption it can be   easily seen that the operator exhibits eigenvalues with finitely supported eigenfunctions that are spread all over the spectrum, compare \cite{V} for a related result on $\Z^{d}$.

The result of this paper stands in clear contrast to results for trees with radial symmetric branching. In \cite{BF} it is shown among other things, that such a tree has absolutely continuous spectrum if and only if the branching function is eventually constant. Hence, if the branching function is random, then the spectrum is almost surely singular. In \cite{HP}, Anderson localization was shown for a corresponding model on quantum graphs. These results rely on the fact that the models are in some sense one dimensional by the radial symmetry.

A similar contrast already occurs for random Schrödinger operators on regular trees. While for small disorder induced by a random potential large parts of the absolutely continuous spectrum are preserved,  \cite{ASW1,AW,FHS2,KLW2, Kl1,Kl3}, the spectrum is almost surely  pure point whenever the disorder is induced by an arbitrary small radial symmetric random potential,  \cite{ASW1}.

Our model is closely related to what is studied in the physics literature under the name quantum percolation. There the randomness of the geometry comes from deleting edges of a given graph independently with some fixed probability. For trees this can be also modeled by a branching process. In particular, the distribution of those branching processes is such that the vertex degree is bounded and vertices have zero forward neighbors with positive probability. While in classical percolation one asks for existence of an infinite cluster, see e.g. \cite{Ly}, in quantum percolation one is concerned with  the transition from localized states to extended states, see e.g. \cite{H1,H2}. In \cite{H1} an asymptotic formula  for a critical probability is derived at which the transition is supposed to happen.

The question of quantum percolation has also attracted some attention in the mathematics community over the past years. For bond percolation on $\Z^{d}$, spectral properties of the Laplacian and spectral asymptotics in the sub- and supercritical regime can be found in \cite{KM,MS1}.
In \cite{V}, the set of eigenvalues with compactly supported eigenfunctions and the integrated density of states of Anderson Hamiltonians for site percolation graphs of $\Z^{^d}$ are studied. Related questions for amenable Cayley graphs are considered in \cite{AV}. A nice introductory survey and a discussion of recent developments can be found \cite{MS} (for further references see also therein).
It is also interesting to note the connection of the spectrum in the subcritical percolation regime and the spectrum of certain lamplighter groups, see \cite{LW}.

\subsection{The model}

Let $T$ be a locally finite rooted tree with root $o$. We do not distinguish in notation between the tree and its vertex set. Likewise we do not distinguish between subgraphs and subsets of vertices.
Let $\ell^{2}(T)$ be the Hilbert space of square summable complex valued functions on the vertices. We study the Laplace operator $\Lp=\Lp({T})$ on
\begin{align*}
    D(\Lp)=\{\ph\in\ell^{2}(T)\mid (x\mapsto \sum_{y\sim x}(\ph(x)-\ph(y)))\in\ell^{2}(T)\}
\end{align*}
acting as
\begin{align*}
    \Lp\ph(x)=\sum_{y\sim x}(\ph(x)-\ph(y)).
\end{align*}
The operator $\Lp$ is selfadjoint and its restriction to the functions of finite support is essentially selfadjoint, see \cite{Woj,KL}. Note that the operator is bounded if and only if the vertex degree is  bounded.

Let  $\A$ be a finite set. We refer to its elements as \emph{labels} or \emph{vertex types}. Let $b$ be a  multi-type Galton Watson branching process with types in $\A$. For background on branching processes see the monographs \cite{AN,H0}. For the sake of brevity, we will speak in the following only of  branching processes meaning  multi-type Galton Watson branching processes.
For a vertex type $j\in\A$ let  $\Te_j=\Te_{j}^{(b)}$ be the subset of rooted trees that are realizations of the process $b$ with root of vertex type $j$. Each tree $\te_{j}\in\Te_j$ is equipped with a \emph{labeling function} $a_j:\te_{j}\to\A,$
assigning to a vertex in $\te_{j}$ its vertex type.
We let $\PP_{j}=\PP_{j}^{(b)}$ be the probability measure on $\Te_j$ induced by $b$.
Furthermore, let $(\Te,\PP)=(\Te^{(b)},\PP^{(b)}) =(\prod_{j\in\A}\Te_j,\bigotimes_{j\in\A}\PP_{j})$. A labeling function $a$ on $\te=(\te_{j})\in\Te$ is defined via its restrictions as
$$a:\te\to\A,\quad a\vert_{\te_j}=a_{j}.$$
We impose the following assumption on the realizations of the branching processes:
\begin{itemize}
  \item [(F)] Each vertex  has at least one forward neighbor.
\end{itemize}
If the realizations of a branching process $b$ satisfy (F) almost surely, then we say $b$ satisfies (F). In particular, this assumption guarantees non extinction of the tree.

For $\te=(\te_{j})\in\Te$, we let $\Lp({\te})$ on $\ell^{2}(\te)$ be the direct sum of the operators $\Lp({\te_{j}})$ on $\ell^{2}(\te_j)$, $j\in\A$.

Let us introduce a distance for multi-type Galton Watson branching processes. We call a function $s:\A\to\N_0$ a \emph{configuration} of forward neighbors of a vertex, i.e.,  $s(k)$ is the number of forward neighbors  with label $k$. Moreover, we denote  $\|s\|=\sum_{k\in\A}s(k)$,  $s\in \N_{0}^{\A}$, which is the  total number of forward neighbors. For a branching process $b$, denote by $\PP^{(b)}_j(s)$ the probability that the configuration of the forward neighbors of a vertex of type $j$ is given by $s\in \N_{0}^{\A}$.
For $p\geq1$, let
$$\W_{p}:=\{b \mbox{ branching process that satisfies (F) and } \max_{j\in\A} \sum_{s\in \N_{0}^{\A}}\PP^{(b)}_{j}(s)\|s\|^{p}<\infty\}.$$
We let $d_{p}:\W_{p}\times \W_{p}\to [0,\infty)$ be the metric given by
$$d_{p}(b_1,b_2)=\max_{j\in\A} \sum_{s\in \N_{0}^{\A}} \big|\PP_{j}^{(b_{1})}(s)-\PP_{j}^{(b_{2})}(s)\big|\|s\|^{p}.$$
If $\inf\{n\in\N \mid \PP_{j}^{(b)}(s)=0\mbox{ for all } s\in \N_{0}^{\A}\mbox{ with } \|s\|> n\mbox{ and all }j\in\A\} <\infty,$ then we say that $b$ has \emph{bounded branching}.

Next, we introduce the deterministic trees to which we compare the random trees. Let a matrix
\begin{align*}
    M:\A\times\A\to\N,\quad (j,k)\mapsto M_{j,k}
\end{align*}
be given. We assume that if $\A$ consists of only one element, then $M$ satisfies $M\geq2$. This condition is necessary in order to exclude the one dimensional situation.
This gives rise to a deterministic branching process $b=b_{M}$ by letting the probability that a vertex of type $j$ has $M_{j,k}$ neighbors of type $k$ be exactly one for all $j,k\in\A$. Then, $\Te^{(b_M)}$ consists of only one element $\T=(\T_{j})_{j\in\A}$. We say the trees $\T_j$ are generated by the substitution matrix $M$ over $\A$. These trees are often referred to as \emph{trees of finite cone type} or \emph{periodic trees}, see \cite{Ly,NW}. As above, $\Lp(\T)$ denotes the direct sum of operators $\bigoplus_{j\in\A}\Lp(\T_{j})$  on  $\ell^{2}(\T) =\bigoplus_{j\in\A}\ell^{2}(\T_{j})$.

The operators $\Lp({\T_{j}})+\be_{o(j)}$, where $\be_{o(j)}=\as{\cdot,\de_{o(j)}}\de_{o(j)}$ and $\de_{o(j)}$ is the delta function of the root $o(j)$, have purely absolutely continuous spectrum which consists of finitely many intervals. For a proof for the case of the adjacency matrix see \cite{KLW1}  and for the general case of label invariant operators see \cite{Kel}.
In particular, positivity of the entries of  $M$ ensures that the spectrum is independent of the label of the root. Moreover, in \cite{KLW1} it is also proven  that the densities of the spectral measure are continuous outside of a finite set $\Sigma_0$.
The operator $\be_{o(j)}$ can be considered as a boundary condition at the root. By the general theory of rank one perturbations, see \cite{Si}, the absolutely continuous spectrum of the operators $\Lp(\T_{j})$ and $\Lp({\T_{j}})+\be_{o(j)}$ coincides.


\subsection{Main result}
The  absolutely continuous spectrum of an operator $H$ is denoted by $\si_{\mathrm{ac}}(H)$.
We will prove the following theorem.

\begin{Thm}\label{main} There exists a finite set $\Sigma_0\subset\R$ such that for all compact $I\subseteq \si_{\mathrm{ac}}(\Lp({\T}))\setminus\Sigma_0$ and $p>1$ there is $\de=\de(I,p)>0$  such that  for all $b\in\W_p$ with $d_{p}(b,b_M)<\de$  the spectrum of $\Lp({\te})$ is purely absolutely continuous in $I$ for almost all $\te\in\Te^{(b)}$.
\end{Thm}

Indeed, $\Sigma_{0}$ is the finite set mentioned above.
If we restrict ourselves to the case of bounded  branching, then we obtain the following immediate corollary.

\begin{Cor}\label{c:main} There exists a finite set $\Sigma_0$ such that for all compact $I\subseteq \si_{\mathrm{ac}}(\Lp({\T}))\setminus\Sigma_0$ there is $q\in(0,1)$ such that for all $b\in\W_{p}$ with bounded branching where the probability for a vertex of type $j$ to have $M_{j,k}$  forward neighbors of type $k$ is larger than $q$ for all $j,k\in\A$, it follows that the spectrum of $\Lp({\te})$ is almost surely purely absolutely continuous in $I$.
\end{Cor}

\begin{Remark}

(a) A special case of our model is the binary tree, where one deletes for each vertex one of the forward neighbors with small probability. In \cite{FHS3}, a proof for preservation of absolutely continuous spectrum is sketched for this case. However, our model does not only allow for a finite number of vertex types, but it is also not restricted to the removal of edges. In particular, the distributions considered in the main theorem do not even assume a bound on the number of forward neighbors of a vertex.

(b) The validity of the theorem does not depend on the choice of the operator $\Delta$. Indeed, the result holds for any selfadjoint label invariant operator in the sense of \cite{Kel}. In particular, it is also true for the normalized Laplacian or the adjacency matrix. However, selfadjointness of the  adjacency matrix is an issue if the vertex degree is unbounded. Nevertheless, in the setting of the corollary  the adjacency matrix is a bounded operator and therefore selfadjoint.

(c) All our estimates are explicit. So, one can get upper bounds for the critical probability  closely related to quantum percolation on a $K$-regular tree.
Let us be more specific: Consider  a $K$-regular tree and let $p\in(0,1)$. For each vertex delete every forward neighbor except for one with  probability $(1-p)$ each. We consider the adjacency matrix at zero energy. We can give upper bounds on the critical probability $p_{K}$  such that for $p>p_K$ some absolutely continuous spectrum is  preserved almost surely.
Since the focus of this paper is primarily on a qualitative result, the estimates are certainly not optimal. In particular, the estimates as given in the proofs yield an upper bound $\sqrt[K-1]{1-2^{-32}K^{-22}/(2K-1)!}$ (nevertheless, with some minor adaptions, see the remark after Proposition~\ref{p:ka}, one can get an upper bound $p_{K}\leq\sqrt[K-1]{1-2^{-33}K^{-22}}$). Since our upper bound tends to $1$ as $K\to\infty$ it does not give information about the behavior expected by physicists. In particular,  \cite{H1} derives an asymptotic formula for the critical probability $p_{K}\sim a/K$, $a\approx 1.42153$,  for the quantum percolation model, where the transition from localized to extended states is supposed to take place.

(d) The theorem deals only with a purely absolutely continuous component of the spectrum. It would be very interesting to know more about the spectrum as a set and the nature of the whole spectrum. Since we deal with rooted trees, the model is not ergodic, so, already non-randomness of the spectrum as a set is an issue.  For the spectrum as a set it is very likely that it stands in a close relationship to the union of  the spectra of trees of finite cone type that are realizations of the distribution. For the nature of the spectrum, one might expect from the perspective of random Schrödinger operators on trees, see \cite{AW}, far more absolutely continuous spectrum than the one of the deterministic tree. Indeed, considering the results of \cite{AW},  one might ask whether singular spectrum  can be excluded almost surely whenever the distributions is close enough to a deterministic one.
\end{Remark}


\subsection{Outline of the proof}

The proof of the theorem is based on techniques developed in \cite{Kel,KLW2} to treat random  operators on trees of finite cone type. These techniques itself are inspired by ideas of \cite{FHS1,FHS2} using hyperbolic geometry and a fix point analysis in order to get bounds on moments of the Green functions. The fundamental difference to the situation in this work is that \cite{Kel,KLW2} deal with small random perturbations appearing everywhere while here we deal with very large perturbations (removal and addition of edges) that occur with small probability.

The Green functions satisfy a recursion relation which we use this to compare the Green functions of the random tree to the ones of the deterministic tree.  Our aim is to bound the difference in the expected value in terms of a function $\gm$ which is related to the standard hyperbolic metric of the upper half plane $\h$ and is given by
\begin{align*}
    \gm(\xi,\zeta)=\frac{|\xi-\zeta|^{2}}{\Im \xi\Im \zeta},\qquad \xi,\zeta\in\h.
\end{align*}

We obtain this bound of the distance of the Green functions by proving a vector inequality: Denote by $\EE\gm$ the vector in $[0,\infty)^{\A}$ where the $j$-th component is the expected value of the $\gm$-distances for the root with label $j$. We prove the vector inequality
\begin{align*}
    \EE\gm\leq(1-\de) P\EE\gm+C,
\end{align*}
with a stochastic irreducible matrix $P$, $\de>0$ and a vector $C$ with positive entries.
For the precise statement see Lemma~\ref{l:VI} in Section~\ref{s:VI}.
The idea is to condition first on the event that (a subset of) the first two spheres of the random tree agrees with the deterministic tree.  By assumption, this event occurs with large probability. For this case, we apply a two step expansion estimate and a uniform contraction estimate from \cite{KLW2}. For the other cases, we prove a different two step expansion estimate, where we use an estimate for  additive perturbations in one argument of $\gm$. This estimate takes the role of the triangle inequality which does not hold for $\gm$. At the end it turns out that the error term from this estimate is compensated by the low probability of these events.

Having the vector inequality above, we conclude by the Perron Frobenius theorem that each component of $\EE \gm$ is bounded. We finish the proof by arguments similar to the ones found in \cite{FHS2,KLW2} and a variant of the limit absorbtion principle.

The paper is structured as follows. In the next section, we start with some preliminaries. In Section~\ref{s:estimates} we provide the crucial estimates such as the two step expansion estimate, the uniform contraction estimate and the vector inequality. These inequalities are used in Section~\ref{s:proof} to prove the main theorem.


\section{Preliminaries}

In this section we recall some basic facts about the Green function on trees. Moreover, we introduce an equivalence relation for the random trees that is determined by the branching in the first two spheres. Finally, we consider conditioned expectations with respect to the equivalence classes and related invariance properties of the Green function.

Let us start with some remarks about notation.
The roots of the components $\te_{j}$ of $\te\in\Te$ are denoted by $o(j)$. By $o'(j)$, we denote a fixed forward neighbor of $o(j)$ of label $j$ whenever such a vertex exists. We denote general rooted trees by $T$ and their roots by $o$. Moreover, we assume that $T$ is always equipped with a labeling function $a$.
We let $\h=\{z\in\C\mid\Im z>0\}$ be the complex upper half plane and always let $z$ be an element of $\h$ which decomposes as $z=E+i\eta$ with $E\in\R$ and $\eta>0$. Moreover, $|\lm|$, for a complex number $\lm$, denotes the modulus of $\lm$, while $|A|$, for a finite set $A$, denotes the cardinality of $A$.


\subsection{Basic facts about the Green function}\label{ss:basic}
Let $\de_{x}$ be the function on a rooted tree $T$ that gives value one to the vertex $x$ and zero to all other vertices and let $\mu_{x}$ be the spectral measure of $\Lp(T)$ with respect to $\de_{x}$. Then, the Green function $z\mapsto G_{x}(z,\Lp(T))$ on $\h$ is given by the Borel transform of $\mu_{x}$, i.e.,
\begin{align*}
    G_{x}(z,\Lp(T)):=\int_{\si(\Lp(T))}\frac{1}{t-z}d\mu_{x} =\as{\de_{x},(\Lp(T)-z)^{-1}\de_{x}}, \quad z\in\h.
\end{align*}
It is well known that the Green function is analytic and maps $\h$ to $\h$ (for background see e.g. \cite{DK}).

As there is a natural ordering of the vertices in a rooted tree given by their distance to the root, there is a natural notion of forward neighbors of vertices. In particular, we denote by $T_{x}$ the forward tree of $x$ and we denote the Green function of the operator $\Lp(T_{x})+\be_{x}$ on $\ell^{2}(T_{x})$ by
\begin{align*}
    \Gm_{x}(z,\Lp(T)):=G_{x}(z,\Lp(T_{x})+\be_{x}),\quad z\in\h,
\end{align*}
where $\be_{x}=\as{\cdot,\de_{x}}\de_{x}$ can be considered as  boundary term since by considering the forward tree the backward edge is 'missing'. (In the case of the adjacency matrix which considered in \cite{KLW1} this is not necessary due to the zero diagonal of the operator.)

It is well known, see for instance \cite{ASW1,Kel,KLW1,Kl1},  that the Green function satisfies the following recursion equation
\begin{align}\label{e:rec}\tag{$\clubsuit$}
-\frac{1}{\Gm_{x}(z,\Lp(T))}=z-\deg_{T}(x)+\sum_{y\in S_{x}^{T}}\Gm_{y}(z,\Lp(T)),\qquad z\in\h,
\end{align}
where $\deg_{T}(x)$ is the vertex degree of $x$ in $T$ if $x\neq o$ and $\deg_{T}(o)$ is the vertex degree of $o$ plus one which is due to adding $\be_{o}$ to $\Lp(T)$. Moreover,  $S_{x}^{T}$ denotes the set of forward neighbors of $x$.
For a branching process $b$ and $\Te=\Te^{(b)}$, we denote the Green functions $z\mapsto\Gm_{x}(z,\Lp(\te))$, $\te\in\Te$, $x\in\te$, for short by
\begin{align*}
    \Gm_{x}^{\te}:=\Gm_{x}(z,\Lp(\te)),\qquad z\in\h,
\end{align*}
and we write $\Gm_{x}^{(\cdot)}$ for the function $\te\mapsto\Gm_{x}^{\te}$.
In the case of trees $\T$ generated by a substitution matrix, we have  $\Gm_{x}(z,\Lp(\T))=\Gm_{y}(z,\Lp(\T))$ for vertices $x,y$ that carry the same label, i.e. $a(x)=a(y)$. We define the vector $\Gm=(\Gm_{j})_{j\in\A}\in\h^{ \A}$ by letting
\begin{align*}
    \Gm_{j}=\Gm_{o(j)}(z,\Lp(\T)),\qquad z\in\h,\,j\in\A.
\end{align*}
Similarly, $\deg_{\T}$ depends only on the label of a vertex and it is given by $1+\sum_{k\in\A}M_{j,k}$ for a vertex of label $j\in\A$. So, we write $\deg(a(x))=\deg_{\T}(x)=1+\sum_{k\in\A}M_{a(x),k}$, for $x\in\T$.
With this notation, the recursion relation for the Green functions $\Lp(\T)$ can be reduced to the finite system of equations given by
\begin{align*}
-\frac{1}{\Gm_{j}} =z-\deg(j)+\sum_{k\in\A}M_{j,k}\Gm_{k},\quad j\in\A.
\end{align*}


In \cite[Theorem 3.1]{Kel} (compare also \cite[Theorem 6]{KLW1}) it is shown that there is a finite set $\Sigma_0$ such that the truncated Green functions $\Gm=(\Gm_{j})$ can be  extended continuously to a function $\Sigma\cup\h\to\h$, where
\begin{align*}
    \Sigma=\si_{\mathrm{ac}}(\Lp(\T)+\be)\setminus\Sigma_0,
\end{align*}
and $\be=\sum_{j\in\A}\as{\cdot,\de_{o(j)}}\de_{o(j)}$ with $o(j)$ being the root of $\T_{j}$, $j\in\A$.
In particular, for all $E\in\Sigma$  and $j\in\A$ the limits  $\Gm_{j}(E,\Lp(\T)) =\lim_{\eta\downarrow0}\Gm_{j}(E+i\eta,\Lp(\T))$
exist, are continuous functions in $E$ and $\Im \Gm_{x}(E,\Lp(\T))>0$. Since the measures $\Im G_{x}(E+i\eta,\Lp(\T))dE$ converge weakly to the spectral measure $\mu_{x}$, we have $\Sigma\subset \si_{\mathrm{ac}}(\Lp(\T_{x})+\be)$ for all $x\in \T$ and, therefore, $\Sigma\subset \si_{\mathrm{ac}}(\Lp(\T)+\be)$ (as $\T_{j}=\T_{o(j)}$).


\subsection{Equivalence classes of random trees}\label{ss:equiv}
For $\te\in\Te$ and $x\in\te$, we  write  $S_{x}^{\te}$ for the set of forward neighbors of $x$ without specifying in which component $\te_{j}$, $j\in\A$, of $\te$ the vertex $x$  is actually contained in.

Let  $\te,\te'\in\Te$.
We say that two subsets $A\subseteq\te$, $ A'\subseteq \te'$ are \emph{isomorphic} if there is a graph isomorphism between $A$ and $A'$ that leaves the labeling invariant. Whenever two sets are isomorphic, we identify the corresponding vertices in notation. For example, for fixed $j\in\A$ all elements of $\Te_{j}$ have a root of label $j$.

For all $\te_{j}$, we fix a vertex $o'(j)$ in $S_{o(j)}^{\te}$ of label $j$  whenever such a vertex exists. In this case, we say that $o'(j)$ exists. Otherwise, when there is no such vertex, we say that $o'(j)$ does not exist and we let $S_{o'(j)}^{\te}=\emptyset$. In $\T_{j}$, the vertex $o'(j)$ always exists  by positivity of the entries of the substitution matrix.

We define
\begin{align*}
S(\te_j):=S_{o'(j)}^{\te}\cup S_{o(j)}^{\te}\setminus \{o'(j)\}.
\end{align*}
Hence, in the case where $o'(j)$ does not exists, we get $S(\te_{j})=S_{o(j)}^{\te}$.

For each $j \in\A$, we define an equivalence relation  on $\Te_{j}$ by
\begin{align*}
    \te_j\cong \te_j'
\end{align*}
whenever $S(\te_{j})$ and $S(\te_{j}')$ are isomorphic. In particular, if  $o'(j)$ exists, then the subsets $S_{o'(j)}^{\te}\cup S_{o(j)}^{\te}\setminus \{o'(j)\}$ and $S_{o'(j)}^{\te'}\cup S_{o(j)}^{\te'}\setminus \{o'(j)\}$ have to be isomorphic and otherwise only  $S_{o(j)}^{\te}$ and $S_{o(j)}^{\te'}$ have to be isomorphic.

We denote the equivalence classes by $[\te_{j}]$, $j\in\A$. Moreover, we write $\te\cong\te'$ if $\te_j\cong\te_j'$ for all $j\in\A$ and  denote the equivalence classes by $[\te]$ correspondingly.

There is a one to one map from the equivalence classes $\{[\te_{j}]\mid \te_{j}\in\Te_j\}$  to $\N_0^{\A}\times\N_0^{\A}$, where
$[\te_{j}]\mapsto(n,m)$ is such that there are $n_{k}$ vertices of label $k$ in $S_{o(j)}^{\te}$ and  $m_{k}$ vertices of label $k$ in  $S_{o'(j)}^{\te}$, $k\in\A$. (In the case, where no vertex of label $j$ exists in $S^{\te}_{o(j)}$, we have $m_{k}=0$  for all $k\in\A$ since $S_{o'(j)}^{\te}=
\emptyset$.) Hence, the set of equivalence classes is countable.

Whenever the dependance of a set or a quantity on  $\te$ is  indeed only a dependance on $[\te]$, we indicate this by replacing $\te$ by $[\te]$ in notation. For example, by the identification of the vertices via labeling invariant graph isomorphisms, the set  $S_{v}^{\te}$  does not depend on the choice of $\te\in[\te']$ for $\te'\in\Te$ for $v\in\{o(j),o'(j)\mid j\in\A\}$. Therefore, we will  write $S_{v}^{[\te]} =S_{v}^{\te}$ and $S([\te_{j}])=S(\te_{j})$. Indeed, $S_{v}^{[\te]}$  depends only on the component of $\te$ where the vertex $v$ lies in. 

Let us stress the importance of this equivalence relation: Let two functions $f:\te\to\C$, $f':\te'\to\C$ for  $\te,\te'\in\Te$ be given. Since the vertex sets of $\te$ and $\te'$ can be totally different except for their roots, it is not clear how to compare these functions. However, if $\te\cong\te'$, then we can compare these functions on $S_{j}(\te)$ and $S_{j}(\te')$, $j\in\A$ by identification of the vertices on these sets. 

\subsection{Conditioned expectations}
For  a measurable set $A\subseteq \Te$ and an integrable function $f$, we write
\begin{align*}
\EE_j(f\mid A)=\frac{1}{\PP_{j}(A)}\int_{A}f(\te_{j} ) d\PP_{j}(\te_{j}),\qquad j\in\A,
\end{align*}
if $\PP_{j}(A)>0$ and $\EE_j(f\mid A)=0$ otherwise. For $A=\Te_{j}$, we write $\EE_{j}(f)=\EE(f\mid \Te_{j})$.

\begin{lemma}\label{l:E}
Let $\te \in\Te$, $j\in\A$ and $x\in S(\te_{j})$. Then, for every function $g$ such that $\te\mapsto g(\Gm^{\te}_{x})$ is integrable
\begin{align*}
\EE_{j}(g(\Gm^{(\cdot)}_{x})\mid[\te_{j}])=   \EE_{o(a(x))}\big(g(\Gm^{(\cdot)}_{o(a(x))})\big)
\end{align*}
and for  $y\in S(\te_{j})$ with $a(x)=a(y)$ and every function $f$ such that $\te\mapsto f(\Gm_{x}^{\te},\Gm_{y}^{\te})$ is integrable
\begin{align*}
\EE_{j}\big(f(\Gm^{(\cdot)}_{x},\Gm^{(\cdot)}_{y})\mid [\te_{j}]\big)=
\EE_{j}\big(f(\Gm^{(\cdot)}_{y},\Gm^{(\cdot)}_{x})\mid [\te_{j}]\big) .
\end{align*}
\end{lemma}
\begin{proof} By \eqref{e:rec} the value of the truncated Green function of a vertex   depends only on the branching of the forward tree. Since the distribution of the branching in a forward tree of a vertex $x$ depends only on $a(x)$, we conclude the first statement. Moreover, the random variables $\Gm_{x}^{\te}$ and $\Gm_{y}^{\te}$ are identically distributed for $a(x)=a(y)$. Furthermore, since  the forward trees of $x,y\in S(\te_{j})$, $x\neq y$, do not coincide, the distribution of their branching is independent. We conclude that
the random variables $\Gm_{x}^{\te}$ and $\Gm_{y}^{\te}$ are independent. Thus, the second statement follows.
\end{proof}


\section{The fundamental inequalities}\label{s:estimates}

In this section we prove the crucial inequalities which are the ingredients for the proof of the theorem.

\subsection{The substitute for the triangle inequality} The distance function $\gm$ on $\h$ defined in the introduction as $\gm(\xi,\zeta)=|\xi-\zeta|^{2}/(\Im \xi\Im \zeta)$, $\xi,\zeta\in\h$, does not satisfy the triangle inequality. Nevertheless, we can give an estimate for additive perturbations in one argument, where the error term does not depend on the other argument.
A similar estimate is found in \cite[Lemma~1]{KLW2}.
\begin{lemma}\label{l:ti}
Let  $\zeta \in\h$, $\lm\in\C$ such that $\zeta+\lm\in\h$. Then, for all $\xi\in\h$,
\begin{align*}
    \gm(\xi,\zeta)\leq c_0(\zeta,\lm)(\gm(\xi,\zeta+\lm)+1),
\end{align*}
where $$c_0(\zeta,\lm)=\ab{1+\frac{\Im \lm}{\Im \zeta}}\ab{1+ \frac{2{|\lm|}}{\Im (\zeta+\lm)}}^{2}.$$
\end{lemma}
\begin{proof} Let
$r=\Im \zeta/\Im(\zeta+\lm)$.
We  distinguish two cases: If $|\xi-(\zeta+\lm)|\geq \Im (\zeta+\lm)/2$, then
$$r\gm(\xi,\zeta)\leq \ab{1+\frac{|\lm|}{|\xi-(\zeta+\lm)|}}^2 \gm(\xi,\zeta+\lm)\leq \ab{1+{\frac{2|\lm|}{\Im (\zeta+\lm)}}}^2 \gm(\xi,\zeta+\lm).$$
If, on the other hand, $|\xi-(\zeta+\lm)|\leq \Im (\zeta+\lm)/2$, then $\Im \xi\geq\Im (\zeta+\lm)/2$. So, we obtain
\begin{align*}
r\gm(\xi,\zeta)&\leq \gm(\xi,\zeta+\lm)+\frac{2|\lm||\xi-(\zeta+\lm)|+|\lm|^{2}}{\Im \xi\Im (\zeta+\lm)}\leq \gm(\xi,\zeta+\lm)+ 2\frac{|\lm|\Im(\zeta+\lm)+|\lm|^{2}}{(\Im (\zeta+\lm))^{2}}.
\end{align*}
The statement now follows from the definition of $c_0$.
\end{proof}
\subsection{The one step expansion estimate}\label{ss:TSE}
In this subsection we prove an inequality that allows to estimate the $\gm$-distance of two Green functions at a vertex by their distances attained at the forward neighbors.



In the following, $g_{x}$ always denotes an arbitrary element of $\h$, but it can be thought as $\Gm_{x}(z,\Lp({T}))$, $x\in T$. Recall that $\Gm_{j}=\Gm_{o(j)}(z,\Lp({\T}))$ are the truncated Green functions of $\Lp(\T)$ indexed by the labels $j\in\A$, where $\T=(\T_{j})$ are the trees generated by a substitution matrix.

For $v\in T$, $x\in S_{v}$, let the weights $q_{x}:\h^{S_{v}}\to(0,1)$ be given by
\begin{align*}
q_{x}(g)&:=\frac{\Im g_{x}}{\sum_{u\in S_{v}}\Im g_{u}}, \qquad g\in\h^{S_{v}},
\end{align*}
and note that $\sum_{x\in S_{v}}q_{x}\equiv1$. Moreover, let $Q_{x,y}:\h^{\{x,y\}}\to[0,1]$ and $\cos\al_{x,y}:\h^{\{x,y\}}\to\R$, $x,y\in T$, be given by
\begin{align*}
Q_{x,y}(g)&:=\frac{\sqrt{\Im g_{x}\Im g_{y}\Im \Gm_{a(x)}\Im \Gm_{a(y)}\gm(g_{x},\Gm_{a(x)})\gm(g_{y},\Gm_{a(y)})}}{\frac{1}{2}(\Im g_{x}\Im \Gm_{a(y)}\gm(g_{y},\Gm_{a(y)})+\Im g_{y}\Im \Gm_{a(x)}\gm(g_{x},\Gm_{a(x)}))},\\
\cos\al_{x,y}(g)&:=\cos\big(\arg(g_{x}-\Gm_{a(x)})\ov{(g_{y}-\Gm_{a(y)})}\big),
\end{align*}
for $g\in\h^{\{x,y\}}$ such that $g_{x}\neq \Gm_{a(x)}$, $g_{y}\neq \Gm_{a(y)}$. Then, $Q_{x,y}$ is the quotient of a geometric and an arithmetic mean which implies that it takes values between $0$ and $1$.
If $g_{x}=\Gm_{x}$ or $g_{y}=\Gm_{y}$, we put $Q_{x,y}(g)=\cos\al_{x,y}(g)=0$.
Sometimes it will be convenient to put a higher dimensional vector such as $g\in\h^{S_{v}}$ into the functions $q_{y}$, $Q_{x,y}$ and $\cos\al_{x,y}$ without indicating explicitly that we actually consider the appropriate restriction of $g$. Hence, by what we discussed so far, we have
\begin{align}\label{e:sum}\tag{$\spadesuit$}
\Big|\sum_{y\in S_{v}}q_{y}(g)Q_{x,y}(g)\cos\al_{x,y}(g)\Big|\leq1,\qquad g\in \h^{S_{v}}.
\end{align}
The functions  $Q_{x,y}$, $\al_{x,y}$ depend  on the unperturbed Green functions $\Gm_{j}$, $j\in\A$, and thus on $z$, but we suppress this dependance on $z$ to ease notation. In any case, by continuity of $\Gm_{j}$ all the estimates that follow will be uniform in $z$ for $z\in I+i(0,1]$ and $I\subset \Sigma$ compact.

For $x\in T$, let $N_{x,j}^{T}$ be the number of forward neighbors of $x$ in $T$ of type $j\in\A$, i.e.,
\begin{align*}
N_{x,j}^{T}=|\{y\in S_{x}^{T}\mid a(y)=j\}|.
\end{align*}
Next, we prove the one step expansion formula. A similar statement can be found in \cite[Lemma~5]{KLW1} and \cite[Lemma~2]{KLW2}.


\begin{lemma}\label{l:OSE} (One step expansion estimate)
Let $I\subset \Sigma$ be compact, $T$ be a rooted tree that satisfies (F) and $v\in T$. If $N_{v,j}^{T}=M_{a(v),j}$ for all $j\in\A$, then, for all $z\in I+i(0,1]$
\begin{align*}
{\gm\ap{\Gm_{v}^{T},\Gm_{a(v)}}}
&\leq \sum_{x\in S_{v}^{T}} \frac{\Im \Gm_{a(x)}}{\sum_{j\in\A}M_{a(v),j}\Im \Gm_{j}} \ab{\sum\limits_{y\in S_v^{T}} q_{y}(\Gm^{T})Q_{x,y}(\Gm^{T})\cos\al_{x,y}(\Gm^{T})} \gm(\Gm^{T}_{x},\Gm_{a(x)}),
\end{align*}
where $\Gm^{T}_{x}=\Gm_{x}(z,\Lp(T))$, $x\in T$.
Otherwise, there exists $c_1>0$ such that for all $z\in I+i(0,1]$
\begin{align*}
{\gm\ap{\Gm_{v}^{T},\Gm_{a(v)}}}
&\leq c_1\Big({\sum_{x\in S_{v}^{T}} \frac{\Im \Gm_{a(x)}}{\sum_{j\in\A}M_{a(v),j}\Im \Gm_{j}} \gm(\Gm_{x}^{T},\Gm_{a(x)})+|S_{v}^{T}|}\Big).
\end{align*}
\end{lemma}
\begin{proof}
We start the proof with two observations:
Firstly,  by a direct calculation one checks
that for $\xi,\zeta,z\in\h$, $c,d\in\R$
\begin{align*}
\gm\Big({-\frac{1}{z-c+\xi},-\frac{1}{z-d+\zeta}}\Big) \leq\gm(\xi-c,\zeta-d)=\gm(\xi,\zeta+c-d).
\end{align*}
Secondly, for  $g\in \h^{ S_{v}^{T}}$, we  calculate directly
\begin{align*}
\Big|{\sum_{x\in S_{v}^{T}}g_x-\sum_{x\in S_{v}^{T}}\Gm_{a(x)}}\Big|^{2}
&=\sum_{x\in S_{v}^{T}}\Im \Gm_{a(x)} \Big({\sum_{x\in S_{v}^{T}} \Im g_{y}\cos \al_{x,y}(g)Q_{x,y}(g)}\Big)\gm(g_{x},\Gm_{a(x)}).
\end{align*}
Combined with \eqref{e:rec}, these observations yield the first statement.\\
For the other statement, let $k=a(v)$. We use \eqref{e:rec} and apply the first observation to get
\begin{align*}
{\gm\ap{\Gm_{v}^{T},\Gm_{a(v)}}}\leq\gm\Big(\sum_{x\in S_{v}^{T}}\Gm_{x}^{T}, \sum_{j\in\A}(M_{k,j}\Gm_{j}+N_{v,j}^{T}-M_{k,j})\Big),
\end{align*}
where we additionally used that $\deg_{T}(v)=1+\sum_{j}N_{v,j}$ and $\deg(k)=1+\sum_{j}M_{k,j}$.
We apply Lemma~\ref{l:ti} with
 $\xi=\sum_{x\in S_{v}^{ T}} \Gm_{x}^{ T}$, $\zeta=\sum_{j\in\A}(N_{v,j}^{T}-M_{k,j} +M_{k,j}\Gm_{j})$ and $\lm=\sum_{j\in\A}(N_{v,j}^{ T}-M_{k,j})(\Gm_{j}-1)$, (clearly $\zeta+\lm\in\h$), to obtain
\begin{align*}
\gm(\Gm_{v}^{ T},\Gm_{v})\leq c_0 \Big({ \gm\Big({\sum_{x\in S_{v}^{ T}}\Gm_{x}^{ T},\sum_{j\in\A}N_{v,j}^{ T}\Gm_{j}}\Big) +1}\Big) \leq c_0\Big({\sum_{x\in S_{v}^{ T}}  \frac{\Im \Gm_{a(x)}}{\sum_{u\in S_{v}^{ T}}\Im \Gm_{a(u)}}  \gm(\Gm^{ T}_{x},\Gm_{a(x)})+1}\Big),
\end{align*}
where we used the second observation and inequality \eqref{e:sum} in the second step. The constant $c_0$ from Lemma~\ref{l:ti} is a product $c_{0}=c_{0}'c_{0}''$ with $c_{0}'=(1+\Im\lm/\Im \zeta)$ and $c_{0}''=(1+2|\lm|/\Im(\zeta+\lm))^{2}$. With our choice of $\zeta$ and $\lm$ the constant $c_{0}$ is given by
\begin{align*}
c_{0}=c_{0}'c_{0}''=\ab{\frac{ \sum_{j}N_{v,j}^{ T}\Im \Gm_{j}} {\sum_{j}M_{a(v),j}\Im\Gm_{j}}} \ab{1+\frac{2|\sum_{j}(N_{v,j}^{ T}-M_{a(v),j})(\Gm_{j}-1)|} {\sum_{j}N_{v,j}^{ T}\Im\Gm_{j}}}^{2}
\end{align*}
and, therefore, since $\sum_{u\in S_{v}^{ T}}\Im \Gm_{a(u)}=\sum_{j\in \A}N_{v,j}^{T}\Gm_{j}$,
\begin{align*}
\gm(\Gm_{v}^{ T},\Gm_{v})\leq c_0''{\sum_{x\in S_{v}^{ T}}  \frac{\Im \Gm_{a(x)}}{\sum_{j\in\A}M_{a(v),j}\Im \Gm_{j}}  \gm(\Gm^{ T}_{x},\Gm_{a(x)})}+c_{0}.
\end{align*}
Let us estimate $c_{0}'$ and $c_{0}''$ to finish the proof. By definition of $\Sigma$ and compactness of $I$ there is $r\geq 1$ such that
$(|\Gm_{k}|+1)/\Im\Gm_{l}\leq r$ for all $k,l\in\A$ and all $z\in I+i(0,1]$.
The first term, $c_{0}'$, can be estimated by
\begin{align*}
c_{0}'={\frac{ \sum_{j}N_{v,j}^{ T}\Im \Gm_{j}} {\sum_{j}M_{a(v),j}\Im\Gm_{j}}}\leq r |S_{v}^{T}|.
\end{align*}
For the second term, $c_{0}''$, we first note that  $\sum_{k}|N_{v,k}^{T}-M_{j,k}|\leq |S_{v}^{T}|+\sum_{k}M_{j,k}-2$ since $T$ and $\T$ have at least one forward neighbor in common  by (F). We estimate
\begin{align*}
\sqrt{c_0''}={1+\frac{2|\sum_{j}(N_{v,j}^{T}-M_{a(v),j})(\Gm_{j}-1)|} {\sum_{j}N_{v,j}^{ T}\Im\Gm_{j}}}\leq 1+2r\ab{1+\frac{\sum_{j}M_{a(v),j}-2}{|S_{v}^{T}|}}\leq 2r\sum_{j\in\A}M_{a(v),j},
\end{align*}
where we used the inequality $1+(m-2)/n\leq m-1$ for $n\geq1$, $m\geq2$ with $n=|S_{v}^{T}|\geq 1$ (by (F)) and $m=\sum_{j}M_{a(v),j}\geq 2$ (by non one-dimensionality) in the last step.
Letting $c_1:=4r^3\max_{k\in\A}(\sum_{j}M_{k,j})^2$, we get the statement.
\end{proof}


\subsection{The two step expansion estimate}\label{s:expansion}
Next, we adapt the one step expansion  to our model to get the two step expansion by iteration.

We first introduce some notation.  
Let the root $o$ of the tree $T$ have label $j\in\A$ and recall that  $S(T)=S_{o'}\cup S_{o}\setminus\{o'\}$  if there is a vertex $o'$ of label $j$ in $S_{o}$ and $S(T)=S_{o}$ otherwise. For $x\in S_{o}^{T}$, let
\begin{align*}
p_{j,x}:=\frac{\Im \Gm_{a(x)}}{\sum_{k\in\A}M_{j,k}\Im \Gm_{k}}
\end{align*}
and for $x\in S_{o'}^{T}$, let
\begin{align*}
p_{j,x}:=
\frac{\Im \Gm_{a(o')}\Im \Gm_{a(x)}}{(\sum_{k\in\A}M_{j,k}\Im \Gm_{k})^{2}},
\end{align*}
where $\Gm_{j}$, $j\in\A$, are the Green functions of $\Lp(\T)$ as above. For $T=\T_{j}$, we  have $\sum_{x\in S({\T_{j}})}p_{j,x}=1$.
Thus,   for $z\in \h\cup\Sigma$, the matrix
$P=P(z):{\A\times\A}\to[0,\infty)$  given by
\begin{align*}
P_{j,k}:=\sum_{{x\in S(\T_{j}),\,}{a(x)=k}} p_{j,x},\qquad j,k\in\A,
\end{align*}
defines a stochastic matrix.

We define the \emph{contraction quantities} $c_{x}:\h^{S^{[\T]}}\to[-1,1]$ for  $x\in S_{o(j)}^{[\T]}$ by
\begin{align*}
c_{x}(g)&:=\sum\limits_{y\in S_{o(j)}^{[\T]}} q_{y}(g)Q_{x,y}(g)\cos\al_{x,y}(g)
\end{align*}
and for $x\in S_{o'(j)}^{[\T]}$ by
\begin{align*}
c_{x}(g)&:=\Big({\sum\limits_{y\in S_{o(j)}^{[\T]}} q_{y}(g)Q_{x,y}(g)\cos\al_{x,y}(g)}\Big) \Big({\sum\limits_{y\in S_{o'(j)}^{[\T]}} q_{y}(g)Q_{o'(j),y}(g)\cos\al_{o'(j),y}(g)}\Big),
\end{align*}
where $g_{o'}$ is defined in the spirit of \eqref{e:rec} as
\begin{align*}
g_{o'(j)}:=-\frac{1}{z-\deg(j)+\sum_{x\in S_{o'(j)}^{\T}}g_{x}}
\end{align*}
and $q_{y}$, $Q_{x,y}$ and $\al_{x,y}$ are taken from the previous subsection.
By inequality \eqref{e:sum}, we see  that $c_{x}$ takes indeed values in $[-1,1]$.

We now come to the two step expansion estimate. For the case $\te_j\in[\T_{j}]$ the following estimate is similar to Proposition~1 of \cite{KLW2}.
In the other cases, we get a similar expression with a multiplicative and additive error term.

\begin{lemma}\label{l:expansion}(Two step expansion estimate) Let $I\subset \Sigma$ be compact and  $j\in\A$. For all $z\in I+i(0,1]$ and  $\te\in[\T_{j}]$,
$${\gm({ \Gm_{o(j)}^{\te},\Gm_{o(j)}})}\leq\sum_{x\in S(\T_j)}p_{j,x}c_{x}(\Gm^{\te}) \gm(\Gm_{x}^{\te},\Gm_{a(x)}).$$
Moreover,  for all $p>1$ there is $c_2>0$ such that for all $z\in I+i(0,1]$ and all $\te\in\Te$
\begin{align*}
\EE_{j}\ab{\gm({ \Gm_{o(j)}^{(\cdot)},\Gm_{o(j)}})^{p}\mid [\te_{j}]}\leq c_2{\big({|S_{o(j)}^{[\te]}|}^{p}+{|S_{o'(j)}^{[\te]}|}^{p}\big)} \Big(\sum_{k\in\A} P_{j,k}\EE_{k}\big({\gm(\Gm_{o(k)}^{(\cdot)},\Gm_{k})^{p}}\big) +1\Big).
\end{align*}
\end{lemma}
\begin{proof}The first statement follows directly by  combining the first statement  of the one step expansion estimate, Lemma~\ref{l:OSE}, for $v=o(j)$ and $v=o'(j)$.

For the second statement let $\te\in\Te$. By the second statement of Lemma~\ref{l:OSE} applied for $v=o(j)$ in the first estimate and for $v=o'(j)$ in the second estimate, we get
\begin{align*}
\EE_{j}\Big(\gm\big(\Gm_{o(j)}^{(\cdot)},\Gm_{o(j)}\big)^{p}\Big| [\te_{j}]\Big) &\leq c_{1}^{p}
\EE_{j}\Big(\Big(\sum_{x\in S_{o(j)}^{[\te]}} \tfrac{\Im\Gm_{a(x)}}{\sum_{k\in\A} M_{j,k}\Im\Gm_{k}} \gm(\Gm_{x}^{(\cdot)},\Gm_{a(x)})+{\big|S_{o(j)}^{[\te]}\big|} \Big)^{p} \Big| [\te_{j}]\Big) \\
&\leq c_{1}^{2p} \EE_{j}\Big(\Big(\sum_{x\in S({[\te_{j}]})}
p_{j,x}\gm(\Gm_{x}^{(\cdot)},\Gm_{a(x)}) +{\big|S_{o(j)}^{[\te]}\big|}+{\big|S_{o'(j)}^{[\te]}\big|}\Big)^{p} \Big| [\te_{j}]\Big).
\end{align*}
Note that in the case, where $o'(j)$ does not exist the second estimate is trivially true.
For $x\in S([\te_{j}])$, let $v(x)=o(j)$ if $x\in S_{o(j)}^{[\te]}$ and $v(x)=o'(j)$ if $x\in S_{o'(j)}^{[\te]}$. 
Since $\sum_{x\in S([\te_{j}])}p_{j,x}\tfrac{M_{j,a(x)}}{N_{v(x),a(x)}^{[\te]}}\leq1$, we get by Jensen's inequality and the inequality $(m+n)^{p}\leq 2^{p}(m^{p}+n^{p})$ for $m,n\geq0$,
\begin{align*}
\ldots\leq 2^{p}c_{1}^{2p}\Big(\sum_{x\in S({[\te_{j}]})}
p_{j,x}\Big(\frac{N_{v(x),a(x)}^{[\te]}}{M_{j,a(x)}}\Big)^{p-1} \EE_{j}\big(\gm(\Gm_{x}^{(\cdot)},\Gm_{a(x)})^{p}\big| [\te_{j}]\big) +\big|S_{o(j)}^{[\te]}\big|^{p}+\big|S_{o'(j)}^{[\te]}\big|^{p}\Big).
\end{align*}
By Lemma~\ref{l:E} and with $c_2:= 2^{p}c_{1}^{2p}$
\begin{align*}
\ldots=c_{2}\Big(\sum_{k\in\A}
\ow P_{j,k} \EE_{k}\big(\gm(\Gm_{o(k)}^{(\cdot)},\Gm_{k})^{p}\big) +\big(\big|S_{o(j)}^{[\te]}\big|+\big|S_{o'(j)}^{[\te]}\big|\big)^{p}\Big),
\end{align*}
where $\ow P_{j,k}:=\sum_{{x\in S([\te_{j}]),\,}{a(x)=k}} p_{j,x}(\frac{N_{v(x),k}^{[\te]}}{M_{j,k}})^{p-1}$, $j,k\in \A$. Since $p_{j,x}$ depends only on $a(x)$ and $v(x)$
\begin{align*}
\ow P_{j,k}=\sum_{{x\in S([\te_{j}]),\,}{a(x)=k}} p_{j,x}\Big(\frac{N_{v(x),k}^{[\te]}}{M_{j,k}}\Big)^{p-1}=\sum_{{x\in S([\T_{j}]),\,}{a(x)=k}} p_{j,x}\Big(\frac{N_{v(x),k}^{[\te]}}{M_{j,k}}\Big)^{p}.
\end{align*}
We finish the proof by estimating $\ow P_{j,k}/ P_{j,k}\leq\sum_{{x\in S([\T_{j}])},{a(x)=k}} ({N_{v(x),k}^{[\te]}})^{p}
\leq {{\big|S_{o(j)}^{[\te]}\big|}^{p}+{\big|S_{o'(j)}^{[\te]}\big|}^{p}}$ since $\sum_{x\in S_{v}^{[\te]}} N_{v,a(x)}^{[\te]}=\big|S_{v}^{[\te]}\big|$ and $\sum_{k}M_{j,k}\geq1$.
\end{proof}


\subsection{The averaged contraction coefficient and the uniform contraction estimate}
In \cite[Proposition~2]{KLW2} a uniform contraction  estimate is proven, which also plays an important role in the proof our main theorem. We recall the corresponding definitions and the statement from  \cite{KLW2}.


\begin{definition}\label{d:Pi}(Label invariant permutations) For $o\in \T_{j}$, we define the set of \emph{label invariant permutations } $\Pi:=\Pi_{j}$ of $S(\T_{j})$ by
\begin{align*}
\Pi:=\{\pi:S(\T_{j})\to S(\T_{j})\mid \mbox{bijective and $a(\pi(x))=a(x)$ for all $x\in S(\T_{j})$ }\}.
\end{align*}
\end{definition}

For $g\in\h^{S(\T_{j})}$  the composition $g\circ\pi$ is of course given as $g\circ\pi=(g_{\pi(x)})_{x\in S(\T_{j})}$. Let  the restriction of the Green function of $\te_{j}$ to $S([\te_{j}])$ be given by $\Gm_{S([\te_{j}])}^{\te}=(\Gm_{x}(z,\Delta(\te)))_{x\in S([\te_{j}])}$.

By Lemma~\ref{l:E}, we have
$$\EE(f(\Gm_{S(\T_{j})}^{(\cdot)})\mid [\T_{j}])=    \EE(f(\Gm_{S(\T_{j})}^{(\cdot)}\circ\pi)\mid[\T_{j}])$$
for  $\pi\in\Pi$ and any function $f$ such that $\te\mapsto f(\Gm_{S(\T_{j})}^{\te})$ is integrable.
The following definition of the averaged contraction coefficient is also taken from \cite{KLW2}.

\begin{definition}(Averaged contraction coefficient) For $p>1$ and $z\in\h$, let the \emph{averaged contraction coefficient} $\ka_j^{(p)}:\h^{S(\T_{j})}\to\R$ be defined as
\begin{align*}
\ka_j^{(p)}(g)&:=\frac{\sum_{\pi\in\Pi}
\ab{\sum_{x\in S(\T_{j})} p_{j,x}c_{x}(g\circ\pi)\gm(g_{\pi(x)},\Gm_{a(x)})}^p} {\sum_{\pi\in\Pi} \sum_{x\in S(\T_{j})}p_{j,x}{\gm(g_{\pi(x)},\Gm_{a(x)})}^{p}}.
\end{align*}
\end{definition}
Note that $z$ enters the definition of $\ka_{j}^{(p)}$ via  the unperturbed Green functions in  the quantities $p_x$, $c_{x}$ and $\gm(g_{x},\Gm_{a(x)})$.

The following proposition is a special case of Proposition~2 of \cite{KLW2}.

\begin{prop}\label{p:ka}(Uniform contraction estimate, \cite[Proposition~2]{KLW2})
Let $I\subset \Sigma$ be compact and $p>1$. There exist $\de_{0}=\de_{0}(I,p)>0$ such that for all $z\in I+i(0,1]$
\begin{align*}
\max_{j\in\A}\sup_{g\in \h^{S(\T_{j})}}\ka_{j}^{(p)}(g)\leq 1-\de_{0}.
\end{align*}
\end{prop}

A consequence of the proposition above is that
$\ka_{j}^{(p)}(\Gm^{\te}_{S([\te_{j}])})\leq 1-\de_{0}$  for all $\te_{j}\in[\T_{j}]$.

\begin{Remark} In the case, where $\T$ is a regular tree, it suffices to consider  a subset of $\Pi$ containing the identity and a permutation that only interchanges two vertices in different spheres. In this way, we can improve the estimate of the constant $p_{K}$ mentioned in Remark~(c) below the main theorem  by replacing $|\Pi|=(2K-1)!$ by $2$.
\end{Remark}


\subsection{The vector inequality}\label{s:VI}

 The goal of this section is to use the two step expansion estimate, Lemma~\ref{l:expansion}, and the uniform contraction estimate, Proposition~\ref{p:ka}, from the previous sections to prove the following vector inequality.

For $p>1$ and  $z\in\h$  denote
\begin{align*}
\EE\gm:=\left(\EE\left( \gm(\Gm_{o(j)}^{(\cdot)},\Gm_{o(j)}^{\T})^{p} \right)\right)_{j\in\A}=\left(\EE\left( \gm(\Gm_{o(j)}(z,\Lp(\te)),\Gm_{o(j)}(z,\Lp(\T)))^{p} \right)\right)_{j\in\A}
\end{align*}
and recall the definition of the stochastic matrix $P$  given by $P_{j,k}=\sum_{{x\in S(\T_{j}),\,}{a(x)=k}} p_{j,x},$ $j,k\in\A$ in Section~\ref{s:expansion}.

\begin{lemma}\label{l:VI}(Vector inequality) Let $I\subset \Sigma$ be compact  and $p>1$. Then there are constants $\eps,\de> 0$ and a vector $C\in[0,\infty)^{\A}$,  such that
\begin{align*}
    \EE\gm\leq(1-\eps) P\EE\gm + C
\end{align*}
for all $z\in I+i(0,1]$ and all $b\in \W_{p}$ with $d_p(b,b_M)<\de$.
\end{lemma}
\begin{proof}
Let $I\subset \Sigma$ be compact and $p>1$. Moreover, let $b$ be a branching process and $\Te=\Te^{(b)}$.
We let $z\in I+i(0,1]$ and $g=(\Gm_{x}(z,\Lp(\te)))_{x\in S(\te_{j})}$ for $\te\in[\T_j]$.
By \eqref{e:rec} and the definition of $g_{o'}$, we have $g_{o'}=\Gm_{o'}(z,\Lp(\te))$. Moreover, $g_x$ and $g_{y}$ are independent for all $x,y\in S(\te_{j})$ and they are additionally identically distributed if  $x$ and $y$  carry the same label. This gives, in particular, $\EE((\sum_{x\in S([\T_{j}])}p_{j,x}c_{x}(g)\gm(g_{x},\Gm_{a(x)}))^p|[\T_{j}]) =\EE((\sum_{x\in S([\T_{j}])} p_{j,x}c_{x}(g\circ {\pi})\gm(g_{\pi(x)},\Gm_{a(x)}))^p|[\T_{j}])$ for all $\pi\in\Pi$.
We use this together with the first part of the two step expansion estimate, Lemma~\ref{l:expansion}, the uniform contraction estimate, Proposition~\ref{p:ka}, and Lemma~\ref{l:E} to compute
\begin{align*}
\EE_{j}\big(\gm(\Gm_{o(j)}^{(\cdot)},\Gm_{j})^{p}\mid[\T_{j}]\big)
&\leq\EE_{j}\Big(\Big(\sum_{x\in S([\T_{j}])}p_{j,x}c_{x}(g)\gm(g_{x},\Gm_{a(x)})\Big)^p\Big|[\T_j] \Big) \\ &= \frac{1}{|\Pi|}\EE_{j}\Big({\sum_{\pi\in\Pi}\Big({\sum_{x\in S([\T_{j}])} p_{j,x}c_{x}(g\circ{\pi})\gm(g_{\pi(x)},\Gm_{a(x)})}\Big)^p\Big|[\T_j]}\Big) \\&= \frac{1}{|\Pi|}\EE_{j}\Big({\ka_{j}^{(p)}(g){\sum_{\pi\in\Pi}\sum_{x\in S([\T_{j}])}p_{j,x}\gm(g_{\pi(x)},\Gm_{a(x)})^p}\Big|[\T_j]} \Big) \\
&\leq(1-\de_{0})\frac{1}{|\Pi|}\sum_{\pi\in\Pi}\sum_{x\in S([\T_{j}])}p_{j,x}\EE_{j} \big(\gm(g_{\pi(x)},\Gm_{a(x)})^p|[\T_j] \big)\\
&=(1-\de_{0})\sum_{x\in S([\T_{j}])}p_{j,x}\EE_{j} \big(\gm(g_{x},\Gm_{a(x)})^p|[\T_j] \big)\\
&=(1-\de_{0})\sum_{k\in\A}P_{j,k} \EE_{k}\big(\gm(\Gm_{o(k)}^{(\cdot)},\Gm_{k})^p\big).
\end{align*}
Let $n_{j}^{[\te]}:={({|S_{o(j)}^{[\te]}|}^{p}+{|S_{o'(j)}^{[\te]}|}^{p})}$. Since the set $\{[\te_j]\mid\te_j\in\Te_{j}\}$ is  countable,  we can sum over it. By the estimate above and the second part of the two step expansion estimate, Lemma~\ref{l:expansion}, we get
\begin{align*}
\lefteqn{\EE_{j}\big(\gm(\Gm_{o(j)}^{(\cdot)},\Gm_{j})^{p}\big) =\PP_{j}([\T_{j}]) \EE_{j}\big(\gm(\Gm_{o(j)}^{(\cdot)},\Gm_{j})^{p}\mid[\T_j]\big)+ \sum_{[\te_j]\neq[\T_{j}]}\PP_{j}([\te_{j}]) \EE_{j}\big(\gm(\Gm_{o(j)}^{(\cdot)},\Gm_{j})^{p}\mid[\te_j]\big)}\\
&\leq \Big((1-\de_{0})\PP_{j}([\T_{j}]) +c_2 \sum_{[\te_j]\neq[\T_{j}]}\PP_{j}([\te_{j}])n_{j}^{[\te]} \Big) \sum_{k\in\A}P_{j,k} \EE_{k}\big(\gm(\Gm_{o(k)}^{(\cdot)},\Gm_{k})^p\big) +c_2\sum_{[\te_j]\neq[\T_{j}]}\PP_{j}([\te_{j}])n_{j}^{[\te]}.
\end{align*}
For $\te_{j}$ let $s_{v}^{[\te]}\in \N_{0}^{\A}$ be the configuration of $S_{v}^{[\te]}$ for $v\in\{o(j),o'(j)\}$, (that is the map $\A\to\N_0$ encoding the number of forward neighbors of each type).  Recall also that $\PP_{j}(s)$ denotes  the probability that a vertex with label $j\in\A$ has the forward neighbor configuration  $s\in \N_{0}^{\A}$.

As the distribution of the branching of the vertices in $S_{o'(j)}$ and $S_{o(j)}\setminus \{o(j)\}$ is independent, we observe that if $o'(j)$ exists in $\te_{j}$, then $$\PP_{j}([\te_{j}])=\PP\big(s_{o(j)}^{[\te]}\big) \PP \big(s_{o'(j)}^{[\te]}\big)$$ and $\PP_{j}([\te_{j}])=\PP_{j}(s_{o(j)}^{[\te]}) $ otherwise.
Using the inequality $m+n\leq 2mn$ for $m,n\geq1$ and the definition of $n_{j}^{[\te]}$ above, we compute
\begin{align*}
\sum_{[\te_{j}]\neq[\T_{j}]}\PP_{j}([\te_{j}])n_{j}^{[\te]}&= \sum_{[\te_{j}]\neq[\T_{j}]} \PP_{j}([\te_{j}]) \Big({\big|S_{o(j)}^{[\te]}\big|}^{p} +{\big|S_{o'(j)}^{[\te]}\big|}^{p}\Big)\\
&\leq 2\sum_{[\te_{j}]\neq[\T_{j}],S^{[\te]}_{o'(j)}\neq\emptyset} \hspace{-.35cm} \PP\big(s_{o(j)}^{[\te]}\big) \PP\big(s_{o'(j)}^{[\te]}\big) \big|S_{o(j)}^{[\te]}\big|^{p} \big|S_{o'(j)}^{[\te]}\big|^{p} + \hspace{-.1cm} \sum_{[\te_{j}]\neq[\T_{j}], S^{[\te]}_{o'(j)}=\emptyset} \hspace{-.35cm}\PP\big(s_{o(j)}^{[\te]}\big) \big|S_{o(j)}^{[\te]}\big|^{p}\\
&= 2\sum_{s,s'\in \N_{0}^{\A},s'\neq s_{o(j)}^{[\T]}} \PP_{j}(s)\PP_{j}(s') \|s\|^{p}\|s'\|^{p}+\sum_{s\in\N_{0}^{\A},s\neq s_{o(j)}^{[\T]}} \PP_{j}(s)\|s\|^{p}\\
&=2\Big(\sum_{s\neq s_{o(j)}^{[\T]}}\big|\PP_{j}(s)\big| \|s\|^{p} \Big)^{2} +\Big(1+\PP\big(s_{o(j)}^{[\T]}\big)\|s_{o(j)}^{[\T]}\|^{p}\Big) \sum_{s\neq s_{o(j)}^{[\T]}}\big|\PP_{j}(s)\big| \|s\|^{p} \\
&\leq 2d_{p}(b,b_M)^{2}+C_{0} d_{p}(b,b_M)
\end{align*}
with $C_{0}=1+(\sum_{k}M_{j,k})^{p}$.
We conclude
\begin{align*}
\EE_{j}(\gm(\Gm_{o(j)}^{(\cdot)},\Gm_{j})^{p})\leq\big(1-\de_{0} +c_{2}(2d_{p}(b,b_M)^{2}+C_0 d_{p}(b,b_M))\big)\sum_{k\in\A}P_{j,k} \EE_{k}\big(\gm(\Gm_{o(k)}^{(\cdot)},\Gm_{k})^p\big) +C.
\end{align*}
and the statement follows with the choice of $\de>0$ such that $\eps:=\de_{0}-c_{2}(2\de^{2}+C_0\de)>0$ and $C:=c_{2}(2\de^{2}+C_0\de)$.
\end{proof}


\section{Proof of the main theorem}\label{s:proof}

Denote by $S_{j}^{\te}(n)$ the set of vertices of $\te_j\in\Te_{j}$ that have distance $n\geq0$ to the root $o(j)$ and let $B_{j}^{{\te}}(n)=\bigcup_{k=0}^{n}S_{j}^{\te}(n)$.

\begin{lemma}\label{l:EEG} Let $I\subset \Sigma$ be  compact and $p>1$. Then, there is $\de> 0$,  such that for all $b\in\mathcal{W}_{p}$ with $d_p(b,b_M)<\de$ and all $n\geq0$
\begin{align*}
\sup_{z\in I+i(0,1]} \EE\Big({\tfrac{1}{|S_{j}^{(\cdot)}(n)|}\sum_{x\in S_{j}^{(\cdot)}(n)} \gm(G_{x}(z,\Lp(\cdot)),\Gm_{a(x)}(z,\Lp(\T)))^{p}}\Big) <\infty,\qquad j\in\A.
\end{align*}
\end{lemma}
\begin{proof}
Let $P=P(z)$ be the stochastic matrix defined in Section~\ref{s:expansion} for $z\in\h\cup\Sigma$. Since the entries of $M$ are positive, so are the entries of $P$. Hence, $P$ is an irreducible stochastic matrix and by the Perron Frobenius theorem there is a positive left eigenvector $u=u(z)\in\R^{\A}$, $z\in I+i[0,1]$, such that $P^{\top}u=u$. We obtain using the vector inequality, Lemma~\ref{l:VI},
\begin{align*}
   \as{u, \EE\gm} \leq(1-\eps) \as{u,P\EE\gm} + C=(1-\eps) \as{u,\EE\gm} + C
\end{align*}
for all $z\in I+i(0,1]$ and all  $b\in\mathcal{W}_{p}$ with $d_p(b,b_M)<\de$, where $\eps,\de>0$ and $C$ are given by  Lemma~\ref{l:VI}.  This implies, in particular, that the $j$-th component of $\EE\gm$ can be estimated by $C/(u_{j}\eps)$ for all $j\in\A$. Since $P$ is continuous on $\h\cup\Sigma$ in $z$, so is $u$. As  $u$ is positive on the compact set $I+i[0,1]$, we have
\begin{align}\label{e:Gm}\tag{$\heartsuit$}
\sup_{z \in I+i(0,1]}\EE(\gm(\Gm_{o(k)}^{(\cdot)},\Gm_{k})^{p})<\infty,\qquad k\in\A.
\end{align}
We now want to use \eqref{e:Gm} to estimate $\gm(G_{x}^{(\cdot)},\Gm_{a(x)})^{p}$ for $x\in S_{j}^{\te}(n)$, $n\geq0$. We will look for a moment at $x$ as the new root of the tree $\te_{j}$. By \eqref{e:rec}, we can expand $G_{y}^{\te}=G_{y}(z,\Lp(\te))$ in terms of its forward neighbors in the tree $\te_{j}$ with respect to the root $x$. For those neighbors of $x$ that are not in $S_{j}^{\te}(n+1)$ we  apply again \eqref{e:rec}. We can do this successively to see that $G^{\te}_{x}$ is a function of  the truncated Green functions $\Gm^{\te}_{y}$, $y\in S_{j}^{\te}(n+1)$.
In this sense, applying the one step expansion estimate, Lemma~\ref{l:OSE},  sufficiently often, we obtain by  estimating all terms $\Im\Gm_{k}/\sum M_{k,l}\Im \Gm_{l}\leq1$
\begin{align*}
\gm(G_{x}^{\te},\Gm_{a(x)})\leq c_{1}^{2n+1}\Big(\sum_{y\in S^{\te}_{j}({n+1})} \gm(\Gm_{y}^{\te},\Gm_{a(y)})+|B_{j}^{\te}(n+1)|\Big).
\end{align*}

We say  $\te_j$ and $\te_{j}'$ are isomorphic up to sphere $n$ if there is a label preserving graph isomorphism between $S_{j}^{\te}({n})$ and $S_{j}^{\te'}({n})$. This defines an equivalence relation on $\Te_{j}$. We denote the equivalence classes by $[\te_{j}]_{n}$.

Now, using the estimate above, we get
\begin{align*}
\lefteqn{\EE_{j}\Big(\tfrac{1}{|S_{j}^{(\cdot)}(n)|} \sum_{x\in S^{(\cdot)}_{j}(n)} \gm(G_{x}^{(\cdot)},\Gm_{a(x)})^{p}\Big) =\sum_{{[\te]}_{n+1}} \tfrac{\PP_{j}({[\te]}_{n+1})}{|S_{j}^{{[\te]}_{n+1}}(n)|} \EE_{j}\Big(\sum_{x\in S^{(\cdot)}_{j}(n)} \gm(G_{x}^{{(\cdot)}},\Gm_{a(x)})^{p} \Big\vert {[\te]}_{n+1}\Big) }\\
&\leq c_{1}^{(2n+1)p} \sum_{[\te_{j}]_{n+1}}\tfrac{\PP_{j}({[\te]}_{n+1})}{|S_{j}^{{[\te]}_{n+1}}(n)|} \EE_{j}\Big(\sum_{x\in S^{(\cdot)}_{j}(n)}\Big( \sum_{y\in S^{(\cdot)}_{j}({n+1})} \gm(\Gm_{y}^{{(\cdot)}},\Gm_{a(y)}) +|B_{j}^{(\cdot)}(n+1)|\Big)^{p}\Big\vert {[\te]}_{n+1}\Big)\\
&=c_{1}^{(2n+1)p} \sum_{[\te_{j}]_{n+1}}\PP_{j}({[\te]}_{n+1}) \EE_{j}\Big( \Big( \sum_{y\in S^{(\cdot)}_{j}({n+1})} \gm(\Gm_{y}^{{(\cdot)}},\Gm_{a(y)}) +|B_{j}^{(\cdot)}(n+1)|\Big)^{p}\Big\vert {[\te]}_{n+1}\Big).
\end{align*}
By  the inequality $(n+m)^{p}\leq 2^{p-1}(n^{p}+m^{p})$ and Jensen's inequality, we proceed
\begin{align*}
\ldots&\leq C \sum_{[\te_{j}]_{n+1}}\PP_{j}({[\te]}_{n+1}){|S_{j}^{[\te]_{n+1}}(n+1)|^{p-1}}
\EE_{j}\Big( \sum_{y\in S_{j}^{(\cdot)}({n+1})}  \gm(\Gm_{y}^{(\cdot)},\Gm_{a(y)})^{p} +|B_{j}^{(\cdot)}(n+1)|^{p}\Big\vert {[\te]}_{n+1}\Big),
\end{align*}
where $C:=2^{p}c_{1}^{(2n+1)p}$. Since the distribution of $\Gm_{y}^{\te}$ for $y\in S_{j}^{\te}(n+1)$ is independent of the number of vertices in $S_{j}^{\te}(n+1)$ and depends only on $a(y)$, we get
\begin{align*}
\ldots&\leq C \EE_{j}\big(|S_{j}^{(\cdot)}(n+1)|^{p}\big)
\sum_{k\in\A}\EE_{k}\Big(\gm(\Gm_{o(k)}^{(\cdot)},\Gm_{k})^{p} \Big) +C\EE_{j}\big(|B_{j}^{(\cdot)}(n+1)|^{p}\big).
\end{align*}
Since $\EE_{j}(|B_{j}^{(\cdot)}(n+1)|^{p})<\infty$ by $b\in \mathcal{W}_{p}$, we conclude the statement of the lemma from \eqref{e:Gm}.
\end{proof}

With Lemma~\ref{l:EEG} at hand, the proof of the main theorem can be obtained using the techniques of \cite{FHS2,KLW2}.

\begin{proof}[Proof of the main theorem] Let $\Sigma_{0}$ be the finite set introduced in Section~\ref{ss:basic}, i.e., $\Gm_{j}$, $j\in\A$, are continuous functions $\h\cup\Sigma\to\h$ with $\Sigma=\si_{\mathrm{ac}}(\Delta(\T)+\beta)\setminus\Sigma_{0}$. By the general theory of rank one perturbations, see \cite[Theorem~II.1]{Si}, the absolutely continuous spectra  of  $\Lp(\te_{j})+\be_{o(j)}$ and $\Lp(\te_{j})$ coincide. Hence, $\Sigma=\si_{\mathrm{ac}}(\Delta(\T))\setminus\Sigma_{0}$.

Let now $I\subset \Sigma$, $p>1$, $b\in\W_{p}$ such that $d_p(b,b_M)<\de$ where $\de>0$ is taken from Lemma~\ref{l:EEG}. Let $j\in\A$ and $\Te=\Te^{(b)}$.
We employ the inequality
$\mo{\xi}\leq 4\gm(\xi,\zeta)\Im \zeta+2|\zeta|$, $\xi,\zeta\in\h,$
from \cite{FHS2} (see inequality (9) in the proof of Theorem~2).
This inequality applied with $\xi=G_{x}(z,\Lp({\te}))$ and $\zeta=\Gm_{a(x)}(z,\Lp(\T))$ combined with Lemma~\ref{l:EEG}  gives for all $n\in\N$
\begin{align*}
\sup_{z\in I+i(0,1]} \EE\Big(\tfrac{1}{|S_{j}^{(\cdot)}(n)|}\sum_{x\in S_{j}^{(\cdot)}(n)}{|G_{x}(z,\Lp({\cdot}))|}^{p}\Big)<\infty,
\end{align*}
since $\Gm_{a(x)}(z,\Lp(\T))$ is bounded on any compact subset of $\Sigma$ as it is continuous on $\Sigma$.

By Fatou's lemma and Fubini's theorem, we obtain
\begin{align*}
\lefteqn{\EE_{j}\Big(\tfrac{1}{|S_{j}^{(\cdot)}(n)|}\sum_{x\in S_{j}^{(\cdot)}(n)} {\liminf_{\eta\downarrow0}\int_{I} {|G_{x}(E+i\eta,\Lp(\cdot))|}^{p}dE}\Big)}\\&\leq \liminf_{\eta\downarrow0}\int_{I} \EE_{j}\Big(\tfrac{1}{|S_{j}^{(\cdot)}(n)|} \sum_{x\in S_{j}^{(\cdot)}(n)}|G_{x}(E+i\eta,\Lp(\cdot))|^{p}\Big)dE\\
&\leq\sup_{z\in{I+i[0,1)}} \EE_{j}\Big({\tfrac{1}{|S_{j}^{(\cdot)}(n)|}\sum_{x\in S_{j}^{(\cdot)}(n)}|G_{x}(z,\Lp(\cdot))|^{p}}\Big) \Leb(I) <\infty.
\end{align*}
Hence, we have almost surely
\begin{align*}
\tfrac{1}{|S_{j}^{\te}(n)|}\sum_{x\in S_{j}^{\te}(n)} \liminf_{\eta\downarrow0}\int_{I}  |G_{x}(E+i\eta,\Lp(\te))|^{p}dE <\infty.
\end{align*}
 Therefore, we have for almost all $\te_{j}\in\Te_{j}$ and for all $x\in\te_{j}$
\begin{align*} \liminf_{\eta\downarrow0}\int_{I} |G_{x}(E+i\eta,\Lp(\te_{j}))|^{p}dE <\infty.
\end{align*}
By a variant of the limiting absorption principle found in  \cite[Theorem~4.1]{Kl3}, (see also \cite[Theorem~2.1]{Si2} and \cite[Theorem~1.4.17]{DK}) the operators $\Lp(\te_{j})$ have purely absolutely continuous spectrum in $I$. \end{proof}

\textbf{Acknowledgements.} The author was inspired by the Workshop 'Structural Probability' in 2008 at the  ESI in Vienna to start this work and he would like to thank the organizers Vadim Kaimanovich and Klaus Schmidt for their generous support. Most valuable comments and suggestions made by Daniel Lenz  on an earlier version of this paper are highly appreciated. The author is indebted to Simone Warzel for  important remarks on the literature and to Ofer Zeitouni for helpful comments. He would also like to thank Wolfgang Spitzer for most stimulating discussions during the Alp Workshop in St. Kathrein in 2009.
While this work was written up the author was enjoying the hospitality of the Hebrew University and he was financially supported by the Israel Science Foundation (Grant no. 1105/10).


\begin{thebibliography}{99}
\bibitem[AN]{AN}  K.B. Athreya, N.E. Ney, \emph{Branching processes}. Die Grundlehren der mathematischen Wissenschaften  196, Springer-Verlag, New York-Heidelberg, 1972.

\bibitem[ASW]{ASW1} M. Aizenman, R. Sims, S. Warzel, \emph{Stability of the absolutely continuous spectrum of random Schr\"odinger operators on tree graphs}, Probability Theory and Related Fields, Volume \textbf{136} (2006),  363-394.

\bibitem[AW]{AW}    M. Aizenman, S. Warzel, \emph{Resonant delocalization for random Schr\"odinger
operators on tree graphs},  to appear in J. European Math. Soc., (2011) arXiv:1104.0969v1.

\bibitem[AV]{AV} T. Antunovi\'c, I. Veseli\'c, \emph{Spectral asymptotics of percolation Hamiltonians on amenable Cayley graphs}. Methods of spectral analysis in mathematical physics,  Oper. Theory Adv. Appl. \textbf{186} (2009), 1--29.

\bibitem[BF]{BF} J. Breuer, R.L. Frank, \emph{Singular spectrum for radial trees}, Rev. Math. Phys.    \textbf{21} (2009), 929--945.

\bibitem[DK]{DK} M. Demuth, M. Krishna \emph{Determining spectra in quantum theory}, Progress in Mathematical Physics \textbf{44}, Birkh\"auser, Boston, 2005.

\bibitem[FHH]{FHH} R. Froese, D. Hasler, F. Halasan, \emph{    Absolutely continuous spectrum for the Anderson model on a product of a tree with a finite graph}, preprint 2010, arXiv:1008.2949.

\bibitem[FHS1]{FHS1} R. Froese, D. Hasler and W. Spitzer,
\emph{Transfer matrices, hyperbolic geometry and absolutely continuous spectrum for some discrete Schrödinger operators on graphs}, J. Funct. Anal. \textbf{230} (2006), 184--221.

\bibitem[FHS2]{FHS2} R. Froese, D. Hasler, W. Spitzer, \textit{Absolutely continuous spectrum for the Anderson model on a tree: a geometric proof of Klein's theorem},   Comm. Math. Phys.  \textbf{269}  (2007), 239--257.

\bibitem[FHS3]{FHS3} R. Froese, D. Hasler, W. Spitzer, \emph{A geometric approach to absolutely continuous spectrum for discrete Schrödinger operators}, in \cite{LSW} (2011), 201--226.


\bibitem[Hal]{Hal} F. Halasan,   \emph{Absolutely Continuous Spectrum for the Anderson Model on Some Tree-like Graphs},  preprint (2009), arXiv:0810.2516.


\bibitem[Har1]{H0} T.E. Harris,  \emph{The theory of branching processes.} Die Grundlehren der Mathematischen Wissenschaften,  \textbf{119} Springer-Verlag, Berlin; Prentice-Hall, Inc., Englewood Cliffs, N.J. 1963.

\bibitem[Har2]{H1} A.B. Harris \emph{1/$\sigma$ expansion for quantum percolation}, Phys. Rev. Lett. \textbf{49} (1982), 486--489.

\bibitem[Har3]{H2} A.B. Harris \emph{Exact Solution of a Model of Localization}, Phys. Rev. Lett. \textbf{49} (1982), 296--299.

\bibitem[HP]{HP} P.D. Hislop and O. Post, \emph{Anderson Localization for radial tree-like random quantum graphs,}  Waves Random Complex Media  \textbf{19}  (2009),  216--261.

\bibitem[Ke]{Kel} M. Keller, \emph{On the spectral theory of operators on trees}, Ph.D. Thesis 2010.

\bibitem[KL]{KL} M. Keller, D. Lenz, \emph{Dirichlet forms and stochastic completeness of graphs and subgraphs}, {J. Reine Angew. Math.}, {to appear}, DOI: 10.1515/CRELLE.2011.122, {arXiv:0904.2985}.

\bibitem[KLW1]{KLW1} M. Keller, D. Lenz, S. Warzel, \emph{On the spectral theory of trees with finite  cone type}, to appear in Israel Journal of Mathematics, (2011) arXiv:1001.3600.

\bibitem[KLW2]{KLW2} M. Keller, D. Lenz, S. Warzel, \emph{Absolutely continuous spectrum for random operators on trees of finite cone type}, to appear in Journal d'Analyse Mathematique, (2011) arXiv:1001.3600.

\bibitem[KM]{KM} W. Kirsch, P. M\"uller, \emph{Spectral properties of the Laplacian on bond-percolation graphs}, Math. Z. \textbf{252}  (2006), 899--916.


\bibitem[Kl1]{Kl1} A. Klein,
\emph{Absolutely continuous spectrum in the Anderson model on the Bethe lattice}, Math. Res. Lett.  \textbf{1} (1994),  399--407.

\bibitem[Kl2]{Kl3} A. Klein, \textit{Extended states in the Anderson model on the Bethe lattice}, Adv. Math. \textbf{133} (1998),  163--184.


\bibitem[KS1]{KlS1} A. Klein, C. Sadel, \emph{{Absolutely continuous spectrum for random operators on the Bethe Strip}}, Mathematische Nachrichten \textbf{285}   (2012), 5--26.

\bibitem[KS2]{KlS2} A. Klein, C. Sadel, \emph{Ballistic Behavior for Random Schrödinger Operators on the Bethe Strip}, J. Spectr. Theory \textbf{1} (2011), 409--442.


\bibitem[LNW]{LW} F. Lehner, M. Neuhauser, W. Woess,
\emph{On the spectrum of lamplighter groups and percolation clusters,} Math. Ann. \textbf{342} (2008), 69--89.

\bibitem[LSW]{LSW} D. Lenz,  F. Sobieczky, W. Woess (eds): \textit{Boundaries and Spectra of Random Walks},  Progress in Probability  \textbf{64}, 2011.


\bibitem[L]{Ly} R. Lyons, \emph{Random walks and percolation on trees}, Ann. Probab. \textbf{18} (1990), 931--958.


\bibitem[MS1]{MS1} P. M\"uller, P. Stollmann, \emph{Spectral asymptotics of the Laplacian on supercritical bond-percolation graphs}, J. Funct. Anal. \textbf{252} (2007), 233--246.

\bibitem[MS2]{MS} P. M\"uller, P. Stollmann, \emph{Percolation Hamiltonians}, in \cite{LSW}, (2011) 23--258.

\bibitem[NW]{NW} T. Nagnibeda, W. Woess, \textit{Random walks on trees with finitely many cone types}. J. Theoret. Probab. \textbf{15}  (2002), 383--422.

\bibitem[S1]{Si} B. Simon \emph{Spectral analysis of rank one perturbations and applications,} Proc. Mathematical Quantum Theory, II: Schrödinger Operators (eds. J. Feldman, R. Froese and L. Rosen), CRM Proceedings and Lecture Notes \textbf{8} (1995),  109--149.

\bibitem[S2]{Si2} B. Simon \emph{Lp norms of the Borel transform and the decomposition of measures}, Proc. Amer. Math. Soc \textbf{123} (1995), 3749-3755

\bibitem[V]{V} I. Veseli\'c, \emph{ Spectral analysis of percolation Hamiltonians},  Math. Ann. \textbf{331} (2005), 841--865.

\bibitem[W]{Woj}  R.K. Wojciechowski,\emph{ Stochastic completeness of graphs}, ProQuest LLC, Ann Arbor, MI, 2008.  Thesis (Ph.D.)--City University of New York.


\end{thebibliography}
\end{document}